\documentclass[11pt]{article}

\usepackage{amssymb}
\usepackage{amscd}
\usepackage{amsmath}
\usepackage{amsfonts}
\usepackage{amsthm}
\usepackage{color}
\usepackage{enumerate}
\usepackage[all,cmtip]{xy}

\usepackage[pdftex]{hyperref}

\theoremstyle{plain}

\newtheorem{thm}{Theorem}[section]
\newtheorem{prop}[thm]{Proposition}
\newtheorem{lemma}[thm]{Lemma}
\newtheorem{cor}[thm]{Corollary}
\newtheorem{defin}[thm]{Definition}
\newtheorem{question}[thm]{Question}

\theoremstyle{definition}

\newtheorem{rem}[thm]{Remark}
\newtheorem{exam}[thm]{Example}

\newcommand{\R}{{\mathbb{R}}}
\newcommand{\N}{{\mathbb{N}}}
\newcommand{\C}{{\mathbb{C}}}
\newcommand{\Z}{{\mathbb{Z}}}
\newcommand{\Q}{{\mathbb{Q}}}

\newcommand{\al}{{\alpha}}
\newcommand{\la}{{\lambda}}

\newcommand{\si}{{\sigma}}

\newcommand{\ga}{{\gamma}}

\newcommand{\om}{{\omega}}
\newcommand{\Om}{\Omega}

\newcommand{\Ci}{{\mathcal{C}}^{\infty}}

\newcommand{\op}{\operatorname}
\newcommand{\con}{\overline}
\newcommand{\bigo}{\mathcal{O}}
\newcommand{\Hilb}{\mathcal{H}}

\newcommand{\Op}{\op{Op}}

\newcommand{\Ham}{\op{Ham}}

\newcommand{\Id}{{\bf 1}}

\newcommand{\cL}{{\mathcal{L}}}
\newcommand{\tHam}{{\widetilde{\rm Ham}}}

\newcommand{\cO}{{\mathcal{O}}}

\newcommand{\cP}{{\mathcal{P}}}

\newcommand{\cU}{{\mathcal{U}}}

\newcommand{\sgrad}{{{ \rm sgrad}}}
\newcommand{\tphi}{{\widetilde{\phi}}}
\newcommand{\tpsi}{{\widetilde{\psi}}}

\begin{document}

\author{Laurent Charles and Leonid Polterovich$^a$}

\footnotetext[1]{Partially supported by the Israel Science Foundation grant 1102/20}

\title{Asymptotic representations of Hamiltonian diffeomorphisms  and quantization}
\maketitle

\begin{abstract}
We show that for a special class of geometric quantizations with ``small" quantum errors, the quantum classical correspondence gives rise to an asymptotic projective unitary representation of the group of Hamiltonian diffeomorphisms. As an application, we get an obstruction to Hamiltonian actions of finitely presented groups.
\end{abstract}

\section{Introduction and main results}
Geometric quantization is a mathematical theory modeling the quantum classical correspondence. The latter is a fundamental physical principle stating that the quantum mechanics contains the classical mechanics in the limit when the Planck constant goes to zero. In the present paper we focus on the correspondence between Hamiltonian diffeomorphisms
modeling motions of classical mechanics, and their quantum counterparts, unitary operators
coming from the Schr\"{o}dinger  evolution. We show that for a special class of geometric quantizations with ``small" quantum errors, which exist on a certain class of
phase spaces (see Theorem \ref{thm-fine}), this correspondence gives rise to an asymptotic unitary representation of the universal cover of the group of Hamiltonian diffeomorphisms
(Theorem \ref{thm-main-0}). Interestingly enough, together with recent results from group theory \cite{DGLT,LO}, this yields an obstruction to Hamiltonian actions of finitely presented groups (Corollary \ref{cor-LO}). Let us pass to precise definitions.

\subsection{Hamiltonian diffeomorphisms}
Let $(M^{2n},\omega)$ be a closed symplectic manifold. Here $\omega$ is a closed differential
$2$-form, whose $n$-th power does not vanish at any point and, thus, gives rise to a volume form on $M$. For a function $f \in C^\infty(M)$ introduce its {\it Hamiltonian vector field}
$\sgrad f$ as the unique solution of the equation $i_{\sgrad f}\omega = -df$. Given a smooth
function $f: M \times [0,1] \to M$, denote $f_t(x):=f(x,t)$, and consider the time dependent vector field $\sgrad f_t$. Its evolution defines a path of diffeomorphisms $\phi_t$ on $M$ with $\phi_0=\Id$. This path is called {\it a Hamiltonian path}, and the diffeomorphisms $f_t$ are called {\it Hamiltonian} diffeomorphisms. The latter form a group denoted by $\Ham(M,\omega)$ (see \cite{P-book} for further details).

Denote by $\tHam(M,\omega)$ the universal cover of $\Ham(M,\omega)$. Its elements  $\widetilde \phi$ are Hamiltonian paths $\{\phi_t\}$, $t \in [0,1]$ with $\phi_0=\Id$, considered up to a homotopy with fixed end points. We write $\phi =\phi_1$ for the projection of $\widetilde \phi$ to $\Ham(M,\omega)$. Every path $\{\phi_t\}$ is uniquely determined by a time-dependent generating Hamiltonian $f_t \in C^\infty(M)$, where the functions $f_t$ are assumed to have zero mean: $\int_M f_t \; \omega^n  = 0$ for all $t$. We shall say that
$\widetilde{\phi} \in \tHam(M,\omega)$ is {\it generated} by a Hamiltonian $f\in C^\infty(M \times [0,1])$.

Let us mention that the fundamental group $\pi_1(\Ham(M,\omega))$ is an abelian group, and we have a central extension
\begin{equation}
\label{eq-sequence}
\xymatrix{
1 \ar[r] &\pi_1(\Ham(M,\omega)) \ar[r]  &\tHam(M,\omega) \ar[r]^\tau & \Ham(M,\omega)\ar[r] &1\;.}
\end{equation}

\subsection{Fine quantizations}
Define a fundamental operation on functions on a symplectic manifold called {\it the Poisson bracket}:
$\{f,g\}= L_{\sgrad f} g$, where $L$ stands for the Lie derivative. We write
$\|f\|=\max |f|$ for the uniform norm of a function $f$.

In what follows we denote by $\cL( \Hilb )$ the space of Hermitian operators
acting on a finite-dimensional complex Hilbert space $\Hilb$, and write
$\mathbb{U}( \Hilb )$ for the unitary group of $\Hilb$.

\begin{defin} \label{def:fine_quantization}
  {\rm  A {\it fine quantization} of $(M,\omega)$ consists of a sequence of positive numbers $\hbar_k$ with $\lim_{k\to \infty} k\hbar_k =1$,
    a family of finite-dimensional complex Hilbert spaces $\Hilb_k$ such that
     \begin{equation}\label{eq-dim}
\dim \Hilb_k = \Bigl(\frac{k}{2\pi}\Bigr)^n  \op{Vol} (M,\omega)+ \bigo(k^{n-1})\;,
\end{equation}
    and a family of $\R$-linear maps $Q_k: C^\infty(M) \to \cL(\Hilb_k)$ with $Q_k(1) =\Id$, satisfying the following properties:
\begin{itemize}
\item[{(P1)}] ({\bf norm correspondence}) $\|Q_k(f)\|_{\op{op}} = \|f\|+ \bigo(1/k)$;
\item[{(P2)}] ({\bf bracket correspondence}) $[Q_k(f),Q_k(g)] = \frac{\hbar_k}{i} Q_k(\{f,g\}) + \bigo(1/k^3)$,
\end{itemize}
where the remainder is understood in the operator norm $\|\cdot\|_{\op{op}}$.}
\end{defin}

The wording ``fine" is chosen in order to emphasize that the remainder in (P2) is $\bigo(1/k^3)$,
as opposed to $\bigo(1/k^2)$, as it happens for a wide class of geometric quantizations.
For K\"{a}hler quantizations (see Section \ref{sec:an-impr-corr} below), the order of the remainder cannot be improved to $\bigo(1/k^4)$, see \cite[p.470]{oim_eq}. It is unknown whether the same holds true for ``abstract" quantizations defined by axioms (P1) and (P2).


Recall that $(M,\omega)$ is {\it quantizable} if the cohomology class $[\omega]/(2\pi)$
is integral. The following conditions on the first Chern class $c_1(TM)$ and the cohomology
class of symplectic form $[\omega]$ of a quantizable symplectic manifold are equivalent:
\begin{itemize}
\item[{(C1)}]  the line $\frac{1}{2} c_1(TM) - \R [\om]$ in $H^2(M, \R)$ intersects
the lattice of integral classes $H^2 (M,\Z) / \text{torsion}$;
\item[{(C2)}] $c_1$ takes even values on $\text{Ker}([\omega])$, where both
  $c_1$ and $[\omega]$ are considered as morphisms $H_2(M,\Z)/\text{torsion}
  \to \R$.
\end{itemize}
Indeed, (C1) yields (C2) immediately. In the opposite direction, choose a basis in $\text{Ker}([\omega])$,  say $e_1,...,e_{m-1}$, and extend it to a basis in $H_2(M,\Z)/\text{torsion}$ by $e_0$.
Then $\omega(e_0) = 2\pi N$, where the number $N \in \Z$ is defined as an integer such that $[\omega]/(2\pi N)$ is a primitive vector. To get (C1) from (C2),  we choose $\lambda = (c_1(e_0)+2p)/(2N)$, with any integer $p$.

\begin{defin} {\rm We say that $(M,\omega)$ satisfies {\it condition }(C) if
it satisfies one of the equivalent conditions (C1) or (C2).}
\end{defin}

Condition (C) may be viewed as a generalisation of the existence of metaplectic
structure. It is more general:  all complex projective spaces satisfy
condition $(\op{C})$ because their second cohomology groups are one-dimensional. However, only the
projective spaces with an odd complex dimension have a metaplectic structure.

\begin{exam}
{\rm Take $M$ to be $\mathbb{C}P^2$ blown up at one point. Let $L, E$ be the
  basis in $H_2(M,\Z)$ with $L$ being the class of a general line and $E$ of
  the exceptional divisor, respectively.  There exist a symplectic forms on
  $M$ with $\omega(L)=2\pi m$, $\omega(E)=2\pi n$, for any integral $m > n > 0
  $ We have $c_1(nL-mE)= 3n-m$, and hence (C2) is satisfied iff $m=n \mod 2
  $.}\end{exam}


\begin{thm} \label{thm-fine}  Every quantizable closed symplectic manifold $M$ satisfying
  condition $(\op{C}))$ admits a fine quantization.
\end{thm}
The proof is given in Section \ref{sec:an-impr-corr}.

\subsection{Asymptotic unitary representation}
Let $Q_k$ be a fine quantization. For a Hamiltonian $f_t$ as above consider
the unitary quantum equation $U_k(t): \Hilb_k \to \Hilb_k$ described by the  Schroedinger
equation
\begin{equation}\label{eq-Schr}
\dot{U}_k(t) = -\frac{i}{\hbar_k} Q_k(f_t)U_k(t), \;\; U_k(0) = \Id\;.
\end{equation}
One can view the time-one map $U_k=U_k(1)$ as a quantization of the element $\widetilde{\phi}$ represented by $f_t$ \cite{L}.

Consider family of maps $\mu:=\{\mu_k\}$,
$$\mu_k : \tHam(M,\omega) \to \mathbb{U}(\Hilb_k),\;\; \widetilde{\phi} \mapsto U_{k}\;.$$
Let us emphasize that $\mu_k(\tphi)$ depends on the specific choice of a Hamiltonian path joining the identity with $\phi$, in the class of paths homotopic with fixed endpoints.

\begin{thm}\label{thm-main-0}$\;$
\begin{itemize}
\item[{(i)}] The unitaries $\mu_k(\tphi)$ and $\mu'_k(\tphi)$ defined via two different choices
of paths homotopic with fixed endpoints representing $\phi \in \tHam(M,\omega)$,
satisfy
\begin{equation}\label{eq-mu-wd}
\|\mu_k(\tphi)-\mu'_k(\tphi)\|_{\op{op}} = \bigo(1/k)\;.
\end{equation}
\item[{(ii)}] For every $\tphi,\tpsi \in \tHam(M,\omega)$
\begin{equation}
\label{eq- almrep}
\|\mu_k(\tphi)\mu_k(\tpsi) -\mu_k(\tphi\tpsi)\|_{\op{op}} = \bigo(1/k)\;.
\end{equation}
\item[{(iii)}] If $\phi \neq \Id$,
\begin{equation}\label{eq-norm-b}
\|\mu_k(\tphi)- \Id \|_{\op{op}} \geq 1/2+ \bigo(1/k)\;.
\end{equation}
\end{itemize}
\end{thm}
The proof is given in Section \ref{sec:proof-theorem-main0}.

\subsection{First constraints on Hamiltonian group actions}
The collection of maps $\mu_k$ gives rise to an interesting
algebraic object. In order to describe it, we need some preliminaries from \cite{DGLT,LO}.
For $p \geq 1$ and an operator $A :\Hilb \to \Hilb$ acting on a $d$-dimensional Hilbert space $\Hilb$
denote by $\|A\|_p$ its $p$-th Schatten norm given by
$$\|A\|_p = \left(\text{tr} \bigl( \bigl(\sqrt{A^*A}\bigr)^p \bigr) \right)^{1/p}\;.$$
Recall that
\begin{equation}\label{eq-Sch}
\|A\|_{\op{op}} \leq \|A\|_p \leq d^{1/p} \|A\|_{\op{op}}\;.
\end{equation}

\begin{defin}[\cite{LO}]\label{def-appr}{\rm
 A group $\Gamma$ is called {\it $p$-norm approximated} if
there exists a family of maps
$$\rho_k: \Gamma \to \mathbb{U}(\Hilb_k)\;,$$
where $\Hilb_k$ is a sequence of Hilbert spaces of growing dimension,
such that
\begin{equation}\label{eq-almostrep-1}
\lim \|\rho_k(x)\rho_k(y) - \rho_k(xy)\|_p = 0, \quad \forall x,y \in \Gamma\;,
\end{equation}
and
\begin{equation}\label{eq-almostrep-2}
\liminf \|\rho_k(x) - \Id\|_p > 0, \quad \forall x \in \Gamma, x\neq 1\;.
\end{equation}
We call any sequence of maps $\rho_k$ satisfying \eqref{eq-almostrep-1} {\it an asymptotic
representation} of $\Gamma$ in the sequence of unitary groups equipped with the $p$-norms.
}
\end{defin}

Theorem \ref{thm-main-0} combined with estimate \eqref{eq-Sch} and formula \eqref{eq-dim}
immediately yields the following result.

\medskip
\noindent
\begin{cor}\label{cor-LO} Assume that a $2n$-dimensional closed symplectic manifold $M$ admits a fine quantization.  Let $\Gamma \subset \tHam(M,\omega)$ be
a finitely presented subgroup with
\begin{equation}\label{eq-Gamma}
\Gamma \cap \pi_1( \Ham (M,\om)) = \{1\}\;.
\end{equation}
Then $\Gamma$ is $p$-norm approximated for every $p >n$.
\end{cor}

Denote by $\cL\cO_p$ the class of finitely presented groups with are {\bf not} $p$-norm
approximated.  Existence of such groups for $p >1$ was established by Lubotzky and Oppenheim in \cite{LO}. For instance, certain finite central extensions of lattices in simple $\ell$-adic Lie groups belong to this class.

Corollary \ref{cor-LO} yields obstructions to actions of groups from $\cL\cO_p$ on certain symplectic manifolds.

\begin{exam}\label{exam-2}{\rm
Let $M$ be a closed oriented surface of genus $\geq 2$ equipped with an area form $\omega$.
Then $\pi_1(\Ham(M,\omega))=1$ (e.g. see \cite{P-book}). Furthermore, $H^2(M,\Z)=\Z$,
and hence $M$ satisfies condition (C) of Theorem \ref{thm-fine}.  Thus no group of class
$\cL\cO_p$ admits a faithful Hamiltonian action on $(M,\omega)$.
}
\end{exam}

Denote by $K_p \subset \pi_1(\Ham(M,\omega))$ the subgroup formed by elements
$\tphi \in \tHam(M,\omega)$ with $\lim_{k \to \infty} \|\mu_k(\tphi)-\Id\|_{p}= 0$.
Assumption \eqref{eq-Gamma} in Corollary \ref{cor-LO}  can be replaced to
\begin{equation}\label{eq-Gamma-1}
\Gamma \cap K_p = \{1\}\;.
\end{equation}
It would be interesting to explore the subgroup $K_p$.

\subsection{Asymptotic projective representations and more constraints}
What can we say about the restriction of the approximate representation $\mu_k$ to the fundamental
group $\pi_1( \Ham (M, \om) ) \subset {\tHam}(M,\omega)$ ? The following enhancement of Theorem \ref{thm-fine}
sheds light on this question.

\begin{thm} \label{thm-fine-1}  Every K\"ahler closed symplectic manifold $M$ satisfying
condition $(\op{C})$ admits a fine quantization which  satisfies
\begin{equation}\label{eq-loop}
\mu_k(\gamma) = e^{ir_k(\gamma)}\Id + \bigo(1/k)\;,
\end{equation}
where $r_k: \pi_1(\Ham(M,\omega)) \to \R/(2\pi\Z)$ is a sequence of homomorphisms.
\end{thm}
\medskip
\noindent
The proof is given in Section \ref{sec-loop}. The homomorphisms $r_k$ will be
explicitly described in terms of action and Maslov invariants. The result
follows from \cite{CLF}, which is developed in the K\"ahler setting. But
there is no serious reason to think that the K\"ahler assumption is essential
here.

\medskip

Denote by $\mathbb{P}\mathbb{U}(\Hilb_k) = \mathbb{U}(\Hilb_k)/S^1$ the projectivization of the unitary group of the Hilbert space $\Hilb_k$. We equip this group with the quotient metric
$\delta_p([A],[B]) = \inf_\theta \|A- e^{i\theta}B\|_p$. Let us state an analogue of Definition \ref{def-appr} for projective representations.

\begin{defin}\label{def-appr-projective}{\rm
 A group $\Gamma$ is called {\it $p$-norm projectively approximated} if
there exists a family of maps
$$\rho_k: \Gamma \to \mathbb{P}\mathbb{U}(\Hilb_k)\;,$$
where $\Hilb_k$ is a sequence of Hilbert spaces of growing dimension,
such that
\begin{equation}\label{eq-almostrep-1-pr}
\lim \delta_p(\rho_k(x)\rho_k(y), \rho_k(xy)) = 0, \quad \forall x,y \in \Gamma\;,
\end{equation}
and
\begin{equation}\label{eq-almostrep-2-pr}
\liminf \delta_p(\rho_k(x),\Id) > 0, \quad \forall x \in \Gamma, x\neq 1\;.
\end{equation}
We call any sequence of maps $\rho_k$ satisfying \eqref{eq-almostrep-1-pr} {\it an asymptotic
projective representation} of $\Gamma$ in the sequence of unitary groups equipped with the $p$-norms.
}
\end{defin}

With this language,  the asymptotic unitary representation $\mu_k$ from Theorem \ref{thm-fine-1} descends to an asymptotic projective representation
$$\nu_k : \Ham(M,\omega) \to \mathbb{P}\mathbb{U}(\Hilb_k),\;\; \phi \mapsto [\mu_k({\tphi})]\;,$$
where $\tphi$ is any lift of $\phi$. Furthermore, every finitely presented subgroup of $\Ham(M,\omega)$ is $p$-norm projectively approximated. The proof is analogous to the one of Theorem \ref{thm-main-0}, with the only extra ingredient being explained in Remark \ref{rem-proj-vsp} below.

Write $\cP\cL\cO_p$ for the class of finitely presented groups which are {\bf not}
$p$-norm projectively approximated. We sum up the previous discussion in the following theorem, which is the main application of our quantization-based technique to group actions on symplectic manifolds.

\begin{thm}\label{thm-main}  Let $(M,\omega)$ be a closed K\"{a}hler manifold of dimension $2n$ with $[\omega]/(2\pi)$ being an integral class and $c_1(TM)$ taking even values on $\text{Ker}[\omega]$. Then every finitely presented subgroup of the group of Hamiltonian diffeomorphisms $\Ham(M,\omega)$ is $p$-norm projectively approximated with any $p > n$. In other words, groups from the class $\cP\cL\cO_p$, $p > n$ do not a admit a faithful Hamiltonian action on $(M,\omega)$.
\end{thm}

\begin{question} Can groups from the class $\cP\cL\cO_p$ act by volume-preserving
diffeomorphisms on closed manifolds?
\end{question}

\subsection{How to construct groups from $\cP\cL\cO_p$? (following \cite{DGLT,LO})}
De Chiffre, Glebsky, Lubotzky, and Thom \cite{DGLT} and Lubotzky and Oppenheim \cite{LO}
came up with a technique leading to examples of groups of the class $\cL\cO_p$ for $p \in (1,+\infty)$. It was explained to us by Lubotzky that the same method shows that these
groups lie in $\cP\cL\cO_p$. The argument from \cite {LO,DGLT} extends {\it verbatim}.
For reader's convenience we provide a brief outline of this argument adjusted to projective
case.

Fix a non-principal ultrafilter $\cU$, and consider the  ultra-product:
$$V_p:= \prod_{n \to \cU} (\text{Mat}(\C, k_n), ||\cdot||_p)\;.$$
Every asymptotic projective representation of $\Gamma$ yields a genuine isometric representation
$\pi_p$ of $\Gamma$ on $V_p$ by conjugation. The crux of the matter is that the action by conjugation is well defined since for $U_1 = e^{i\theta}U_2$, we have $U_1AU_1^* = U_2AU_2^*$.

Given a class of groups $\cP$, we say that a group $\Gamma$ is {\it residually $\cP$} if
for every element $x \in \Gamma \setminus {\Id}$ there exists a homomorphism
from $\Gamma$ to a group from $\cP$ whose kernel does not contain $x$. Interesting classes
of groups include {\it linear groups} (those, admitting a faithful finite-dimensional representations) and finite groups.

\begin{prop}[\cite{DGLT}] \label{prop-dglt}  Let $\Gamma$ be a finitely presented group with the following properties:
\begin{itemize}
\item[{(a)}] $H^2(\Gamma,\pi_p)=0$;
\item[{(b)}] $\Gamma$ is not residually linear.
\end{itemize}
Then $\Gamma \notin \cP\cL\cO_p$.
\end{prop}
Indeed, assumption (a) enables one to apply a Newton-type process which yields a genuine
representation of $\Gamma$ on $\text{Mat}(\C, k_n)$ for almost all $n$ with respect to the ultra-filter. Moreover, every $x \neq \Id$ does not lie in its kernel for almost all $n$.
But this contradicts assumption (b).

The group $\Gamma$ is constructed in two steps:
\begin{itemize}
\item [{(i)}] Take a cocompact lattice $\Gamma_0$ in a simple Lie group $G$
of rank $\geq 3$ over $\ell$-adic numbers with $\ell$ sufficiently large;
\item [{(ii)}] Take a special finite central extension $\Gamma$ of $\Gamma_0$
which is not residually finite (Deligne).
\end{itemize}

The paper \cite{DGLT} proposes a specific example of the lattice $\Gamma_0$,
$$\Gamma_0 = \mathbb{U}(2n) \cap Sp(2n,\Z[\sqrt{-1},1/p])\;$$
considered as a cocompact lattice in $Sp(2n,\Q_p)$.

The central extension $\Gamma \to \Gamma_0$, based on technique of Deligne, is quite complicated,
and we refer to \cite{DGLT} for details.

\medskip

In order to verify assumption (a) of Proposition \ref{prop-dglt}, the following features are used: first, the Lie group $G$
acts on a special simplicial complex (a Bruhat-Tits building); here one uses the $\ell$-adic nature of
the situation. Second, the representation $\pi_p$ is a particular case of an isometric
representation on Banach spaces from a special class: they are obtained from Pisier's $\theta$-Hilbertian spaces (where $\theta$ depends on $p$) by using quotients, $l_2$-sums and ultra-products.

For verifying assumption (b) of Proposition \ref{prop-dglt}, one uses the (immediate consequence of) Malcev Theorem: any residually linear group is residually finite. This completes our outline of the argument from \cite{DGLT,LO}.

\subsection{Stability}
Another application of Theorem \ref{thm-main-0} deals with the following stability
question: given a subgroup $\Gamma \subset \tHam(M,\omega)$, is its quantization
$\mu_k|_\Gamma: \Gamma \to \mathbb{U}(\Hilb_k)$ close to a genuine representation?
It follows that the answer is affirmative for the class of $p$-norm stable groups
defined as follows \cite{LO,DGLT}. Here we include the case $p = \infty$, i.e. of the operator norm. Let $\Gamma$ be a finitely presented group defined by finite collections of generators $S$ and relations $R$, considered as subsets of the free group $\mathbb{F}_S$
generated by $S$. The {\it $p$-norm stability} means that for every $\epsilon >0$ there exists $\delta >0$ such that for every finite-dimensional Hilbert space $\Hilb$ and every homomorphism
$t: \mathbb{F}_S \to \mathbb{U}(\Hilb)$ with $$\max_{r \in R} \|t (r)-\Id\|_p \leq \delta\;,$$ there exists a homomorphism $\rho: \Gamma \to \mathbb{U}(\Hilb)$ whose lift $\overline{\rho}: \mathbb{F}_S \to \mathbb{U}(\Hilb)$ satisfies
$$\max_{s \in S} \|t(s) - \overline{\rho}(s)\|_p < \epsilon\;.$$
Let us mention that all finite groups are operator norm stable by  \cite{Grove, K}.

\medskip
\noindent
\begin{cor}\label{cor-LO-1} Assume that a $2n$-dimensional closed symplectic manifold $M$ admits a fine quantization.  Let $\Gamma = \langle S|R\rangle \subset \tHam(M,\omega)$ be
a finitely presented $p$-norm stable subgroup, where $p >n$. There exists a family
of homomorphisms $\rho_k: \Gamma \to \mathbb{U}(\Hilb_k)$ such that
$$\max_{s \in S} \|\mu_k(s) - \rho_k(s)\|_p \to 0,\;\; k \to \infty\;.$$
\end{cor}

\medskip
\noindent\begin{rem}{\rm Some examples of finite subgroups of $\tHam(M,\omega)$ come from the following construction. Let $F \subset \Ham(M,\omega)$ be a finite group acting in a Hamiltonian way on a closed quantizable symplectic manifold  $(M,\omega)$. For instance, any unitary representation
of $F$ on a finite-dimensional complex Hilbert space $V$ yields an action of $F$
on the projectivization $\mathbb{P} (V)$. Denote by $\widetilde{F} \subset \tHam(M,\omega)$
as the full lift of $F$. If $F$ is perfect, there exists a finite abelian
extension $G$ of $F$, called the universal extension \cite{Rosenberg} such that the
following diagram commutes:
\[\xymatrix{
G \ar[d] \ar[r] &F\ar[d]^\Id \\ \widetilde{F} \ar[r]^\tau &F
} \]
This provides a monomorphism of $G$ into $\tHam(M,\omega)$.
}
\end{rem}

Let us note also that for any finite subgroup $F \subset \Ham(M,\omega)$, the restriction
of the asymptotic projective representation $\nu_k$, which we constructed for quantizable K\"{a}hler manifolds satisfying condition (C), the restriction $\nu_k |_F$ is close to a genuine projective representation, see \cite{Grove}.

\subsection{Bibliographical and historical remarks}
A few bibliographical remarks are in order. For K\"{a}hler quantization with
metaplectic correction an asymptotic representation of the quantomorphisms group
of a prequantum circle bundle over a closed symplectic manifold is constructed by Charles in
\cite{oim_eq}. In the present paper we generalize this result in two directions: first,
we prove it for arbitrary fine quantizations, and second, for K\"{a}hler quantization,
we impose Condition (C) instead of the assumption that the canonical bundle admits a
square root.

Charles showed that quantization enables one to reconstruct Shelukhin's quasi-morphism on $\tHam(M,\omega)$ \cite{C-q}. Ioos, Kazhdan and Polterovich explored a link between quantization and almost representations of Lie algebras \cite{IKP}.

Constraints on smooth actions of finitely presented groups on closed manifolds
is a classical and still rapidly developing subject. Its highlight is Zimmer's famous conjecture \cite{Zimmer} which, roughly speaking, states that higher rank lattices
in semisimple Lie groups cannot act on manifolds of sufficiently small dimension.
This conjecture was recently resolved in a breakthrough work by Brown, Fisher, and Hurtado
\cite{BFH}. Some results on Hamiltonian actions were obtained by Polterovich, Franks and Handel. We refer to Fisher's survey in \cite{Zimmer} for a more detailed discussion. It would
be interesting to explore potential actions of the group constructed in \cite{DGLT, LO} and described above, which is a finite extension of a higher rank $\ell$-adic lattice with sufficiently high $\ell$, along the lines of \cite{BFH}. As we have learned from David Fisher,
this problem is at the moment open. Furthermore, Fisher conjectured existence of constraints
on actions of such groups.

\section{Constructing fine quantizations}  \label{sec:an-impr-corr}
In this section we prove Theorem \ref{thm-fine}  by constructing a fine quantization,
which will be denoted by $\Op_k$.

In the usual Toeplitz-K\"ahler quantization, we consider a compact K\"ahler manifold
$(M,\om)$ equipped with a holomorphic Hermitian line bundle $L$ whose Chern
connection has curvature $\frac{1}{i} \om$. The quantum space is defined as
the space $\Hilb_k$ of holomorphic sections of $L^k\otimes L'$, where $L' $ is
an auxilliary  Hermitian holomorphic line bundle. Here, the parameter $k$
is a positive integer. The large $k$ limit is the {\em semiclassical limit} where
in first approximation the quantum mechanics reduces to the classical
mechanics of $M$ considered as the classical phase space. In this context, a
standard way to define a quantum observable from a classical one is the {\em Berezin-Toeplitz} quantization:  for any $f \in \Ci (M, \R)$, we let $T_k (f)$ be the endomorphism of $\Hilb_k$ such that
\begin{gather} \label{eq:def_Toep}
  \langle T_k (f) \psi, \psi' \rangle = \langle f \psi, \psi' \rangle
\end{gather}
for any
$\psi, \psi' \in \Hilb_k$. Here the scalar product of $\Ci ( M, L^k \otimes L')$ is  given by
integrating the pointwise scalar product against the Liouville volume form.

The basic properties of these operators are the following equalities which
holds for any $f,g \in \Ci (M)$
\begin{xalignat}{2} \label{eq:old}
\begin{split}
  T_k  (
  fg) & = T_k ( f) T_k (g) + \bigo ( k^{-1})  \\
  [T_k
  (f) , T_k  (g) ] & = (ik)^{-1} T_k ( \{ f,g \} ) + \bigo (
  k^{-2}) \\
  \op{tr}
  (T_k ( f)) & = \Bigl(\frac{k}{2\pi}\Bigr)^n \int _M f \mu + \bigo ( k^{n-1})
\end{split}
\end{xalignat}
see \cite{BoGu}, \cite{BoMeSc}.
Furthermore $\| T_k(f) \|_{\op{op}} = \| f\| + \bigo (k^{-1})$. Observe that in the
bracket correspondence (second line of \eqref{eq:old}), the remainder is a $\bigo (k^{-2})$, so we miss the
fine quantization condition given in Definition \ref{def:fine_quantization}.

The first order correction to \eqref{eq:old} have been computed in \cite{oim_eq}. Introduce
for any $f \in \Ci (M)$, the operator
\begin{gather} \label{eq:def_op_k}
\Op_k  ( f) := T_k ( f - (2 k )^{-1} \Delta f )
\end{gather}
where
$\Delta$ is the holomorphic Laplacian of $M$ (in complex coordinates $\Delta
f = \sum G^{ij} \partial_{z_i} \partial_{\con{z}_j}$ with $(G^{ij})$ the
inverse of $(G_{ij})$ given by $\om = i \sum G_{ij} dz_i \wedge dz_j$). Since $\op{Op}_k ( f ) = T_k (f) + \bigo ( k^{-1})$, the operators
$\op{Op}_k (f)$ satisfy \eqref{eq:old} as well. The novelty is that we have now
some explicit formulas for the first corrections
\begin{xalignat}{2} \label{eq:new}
  \begin{split}
 \Op_k (f ) \Op _k ( g )  & =     \Op_k ( fg )   + \frac{i}{2k}  \Op_k ( \{ f, g \})
 + \bigo ( k^{-2})  \\
 [ \Op_k ( f), \Op_k (g) ] & = (i k)^{-1} \Op_k ( \{ f,g\} - k^{-1} \om_1 ( X_f, X_g) ) + \bigo ( k^{-3})\\
\op{tr}( \Op_k ( f) ) & = \Bigl( \frac{k}{2 \pi} \Bigr) ^n \int_M f ( \om +
k^{-1} \om_1)^n/ n! + \bigo( k^{-2})
\end{split}
\end{xalignat}
see \cite{oim_eq}. Here $\om_1 = i  ( \Theta' -  \frac{1}{2} \Theta_K) $ where $\Theta'$ and $\Theta_K$ are the
Chern curvature of $L'$ and the canonical bundle $K$ respectively. In complex
coordinates as above,  $\Theta_K = \partial \con{\partial} \ln \det
( G_{ij} )$

In the case where $M$ has a metaplectic structure, one can choose for $L'$ a
square root of the canonical bundle, so that $\om_1 =0$ and we get our fine
quantization. More generally, to prove the existence of fine quantizations
under assumption $(\op{C})$, we construct a convenient auxiliary bundle $L'$.

\begin{lemma} Assume condition $(\op{C})$.  Then there exists a holomorphic
  Hermitian line bundle $L'$ such that $\om_1 = \la \om $  with $\la \in \Q$.
\end{lemma}

\begin{proof}
The basic observation we need is that for any line bundle $D$ and integer $m$ such that $D^m$
is equipped with a Hermitian and holomorphic structures,
$D$ has natural holomorphic and Hermitian structures inducing the ones of
$D^m$.  Furthermore the Chern curvature of $D$ is $1/m$ times the Chern
curvature of $D^m$.

Now, the assumption that $\frac{1}{2}c_1^{\R} (K) + \R [\om]$
  intersects the lattice of integral classes means that there exists a line
  bundle $L'$ such that $c_1 ^{\R} (L') =  \frac{1}{2} c_1 ^{\R} ( K ) + \la
  c_1^{\R} ( L)$. Since $c_1^{\R} (L) \neq 0$, $\la = p/q$ is rational. So
  $(L' )^{2q} = K^q \otimes L^{2 p}\otimes T$ where $T$ is a torsion line
  bundle, i.e. $T^m = 1$ for some $m
  \in \N$. We endow $T$ with the holomorphic and Hermitian structures such
  that $T^m$ becomes the trivial Hermitian and holomorphic line bundle, so
 that the Chern curvature of $T$ is zero. Then we endow $L'$ with the
 Hermitian and holomorphic structure compatible with the isomorphism $(L') ^{2q}
 = K^q \otimes L^{2 p}\otimes T$. So the Chern curvature $\Theta'$, $\Theta$
 and $\Theta_K$ of $L'$, $L$ and $K$ satisfy $\Theta' = \frac{1}{2} \Theta_K +
 \la \Theta$. So $\om_1 = i \la \Theta = \la \om$.
\end{proof}

In the case where $\om_1 = \la \om$, the second and third equations of
\eqref{eq:new} reads
\begin{gather} \label{eq:new_corresp_hbar}
\begin{split}
 [ \Op_k ( f), \Op_k (g) ] = (i(k+ \la))^{-1}  \Op_k ( \{ f,g\} ) + \bigo ( k^{-3})\\
\op{tr}( \Op_k ( f) ) = \Bigl( \frac{k+ \la}{2 \pi} \Bigr)^{n} \int_M f   \mu  + \bigo( k^{n-2})
\end{split}
\end{gather}
which proves Theorem \ref{thm-fine} for a K\"ahler manifold with $\hbar_k =
(k+ \la)^{-1}$.

Let us generalize
this to symplectic manifolds.
So we start
with a symplectic compact manifold $(M, \om)$ such that $\frac{1}{2\pi} [\om]$
is integral. We introduce a Hermitian line bundle $L$ with
Chern class $\frac{1}{2\pi} [\om]$ and a second Hermitian line bundle $L'$.
We denote by $\Om_1 \in H^2 (M, \R)$ the cohomology class $$ \Om_1  =  \tfrac{1}{2 \pi} \bigl( c_1^{\R} ( L') - \tfrac{1}{2} c_1 ^{\R} ( K) \bigr).$$
Here, the canonical bundle $K$ is defined through any almost complex structure
compatible with $\om$. It is well known that the Chern class of $K$ only
depends on $\om$. If
$\Hilb_k$ is a finite dimensional subspace of $\Ci ( M, L^k \otimes
L')$, we can define as before the Toeplitz operators $T _k (f)$ by \eqref{eq:def_Toep}. Then
we have the following results:
\begin{enumerate}
  \item by \cite{oim_symp}, cf. also \cite{BoGu}, \cite{MaMa}, one can choose the family $(\Hilb_k )$ so that the operators $T_k  (f)$ satisfy \eqref{eq:old}.
  \item by \cite{oim_sub}, there exists a real differential operator $P : \Ci
    (M) \rightarrow \Ci (M)$ such that $\Op_k(f) = T_k ( f)  +
    k^{-1}T_k ( Pf)$ satisfies \eqref{eq:new} with $\om_1$ a
    representative of $\Om_1$. Furthermore, by adding to $P$ a vector field, one modifies $\om_1$ by an exact form. Choosing conveniently this vector field, we can obtain any representative of $\Om_1$.
  \end{enumerate}
If condition $(\op{C})$ holds, we can choose $L'$ so that $\Om_1 = \la [\om]$ for some $\la \in \Q$. Choosing
$P$ so that $\om_1 = \la \om$, we obtain equations \eqref{eq:new_corresp_hbar}.

\section{Quantum dynamics}

\subsection{The Egorov theorem for fine quantizations}
We start with the Egorov theorem for fine quantizations.
Let $f_t$ be a classical Hamiltonian generating the Hamiltonian flow $\phi_t$, and let $U_k(t)$ be the corresponding quantum evolution.

\begin{thm} \label{thm-Egorov}
For every function $g \in C^{\infty}(M)$
\begin{equation}\label{eq-Egorov}
\|Q_k(g \circ \phi^{-1}) - U_k Q_k(g) U_k^{-1}\|_{\op{op}} = O (\frac{1}{k^2})\;,
\end{equation}
where the remainder depends on $f$ and $g$.
\end{thm}

This formula readily follows from  \cite[Proposition 2.7.1]{L}.
Let us emphasize that the quantum map $U_k$ depends
on the Hamiltonian $f$ generating the diffeomorphism $\phi$.  This dependence will be analyzed later.

\medskip
\noindent
{\bf Proof of the Egorov theorem \eqref{eq-Egorov}:}

Recall that if $\phi_t$ is the Hamiltonian flow generated by a time-dependent
Hamiltonian $f_t(x)$, the flow $\phi_t^{-1}$ is generated by $\bar{f}_t:= -f_t \circ \phi_t$.
It follows that for any function $g \in C^{\infty}(M)$
\begin{equation}\label{eq-evolution}
\frac{d}{dt} g\circ \phi^{-t} = (\phi^{-t})^* (L_{\sgrad \bar{f}_t} g)= (\phi^{-t})^*\{\bar{f}_t,g\}=
-\{f_t,g\circ \phi^{-t}\}\;.
\end{equation}

Next, turn to the analysis of the Schr\"odinger equation
$\dot{\xi}= -\frac{i}{\hbar_k} Q_k(f_t)\xi$.
Introduce the family of unitary operators
$$U(s,t): \Hilb_k \to \Hilb_k\;, \xi(s) \mapsto \xi(t)$$ which sends
the solution at time $s$ to the solution at time $t$.
Observe that $U(0,t)=U_k(t)$ is the Schr\"odinger  evolution,
$U(t,t) = \Id$ and $U(s,t) = U(t,s)^{-1}=U(t,s)^*$. The Schr\"odinger equation yields
\begin{equation} \label{eq-schr-deriv}
\frac{\partial}{\partial s} U(t,s) = -\frac{i}{\hbar_k}Q_k(f_s) U(t,s)\;, \frac{\partial}{\partial s} U(s,t) =-\frac{i}{\hbar_k} U(s,t)Q_k(f_s)\;.
\end{equation}

Put now $B(s) := U(s,1)Q_k(g \circ \phi_s^{-1})U(1,s)$, so that
$B(0) = U_kQ_k(g )U_k={-1}$ and $B(1)=Q_k(g\circ  \phi_1^{-1} )$.
From \eqref{eq-evolution} and \eqref{eq-schr-deriv}
we get that
$$\frac{dB}{ds} = U(s,1)\left( \frac{i}{\hbar_k} [Q_k(f_s),Q_k(g \circ \phi_s^{-1})]- Q_k(\{f_s,g\circ \phi^{-s}\}\right)U(1,s)\;.$$
Observe that the functions $f_s$ and $g \circ \phi_s^{-1}$, $s \in [0;1]$ form a compact family
with respect to $C^{\infty}$-topology, and hence by bracket correspondence (P2)
$\max_s\|dB/ds\|_{\op{op}} = \bigo(1/k^2)$. Thus
$$\|Q_k(g \circ \phi^{-1}) - U_k Q_k(g) U_k^{-1}\|_{\op{op}}= \|\int_0^1 dB/ds(s)\; ds \|_{\op{op}}= \bigo(1/k^2)\;,$$
as required. \qed

\subsection{Proof of Theorem \ref{thm-main-0}} \label{sec:proof-theorem-main0}

Suppose that we have two Hamiltonian paths $\gamma_0= \phi_{0,t}$ and $\gamma_1= \phi_{1,t}$,
$t \in [0;1]$ with $\phi_{0,1}=\phi_{1,1} = \phi$, which are homotopic with fixed end points
through a family $\phi_{t,s}$, $s \in [0,1]$.  Denote by $U_k(\phi_{1,j})$ the time one map of the Schroedinger evolution obtained by the quantization of $\gamma_j$. We claim that
\begin{equation}\label{eq-homotopy}
\|U_k(\phi_{1,1})-U_k(\phi_{1,0})\|_{\op{op}} = \bigo(1/k)\;.
\end{equation}
To see this, look at the family $\phi_{t,s}$ and denote by $p_{t,s}$ the generating Hamiltonian
when $s$ is fixed, $t$ varies, and by $q_{t,s}$ the Hamiltonian when $t$ is fixed, $s$ varies.
All the Hamiltonians are assumed to have zero mean. Then
\begin{equation}\label{eq-partial}
\partial_s p = \partial_t q + \{p,q\}\;.
\end{equation}

Put $A=\hbar_k^{-1}Q_k(p)$, $C=\hbar_k^{-1}Q_k(q)$. Let $U(t,s)$ be the unitary evolution of
$$\partial_t U = -iAU\;$$ with $U(0,s) =\Id$. Note that
$$
U_k(\phi_{1,1})= U(1,1),\;\;U_k(\phi_{1,0}) = U(1,0)\;.$$
Define $B$ by
\begin{equation}\label{eq-U}
\partial_s U = -iBU\;.
\end{equation}
Then
$$\partial_s \partial_t U = -iA\partial_s U - i \partial_s A U = -iABU -i \partial_s A U \;,$$
$$\partial_t \partial_s U = -iB\partial_t U - i \partial_t B U= -iBAU - i \partial_t B U\;.$$
Subtracting and rearranging we get
$$\partial_t B  = \partial_s A  -i[A,B]\;.$$
Further, by \eqref{eq-partial}
$$\partial_t C = \hbar_k^{-1}Q_k(\partial_t q) = \hbar_k^{-1}Q_k(\partial_s p) + \hbar_k^{-1}Q_k (\{p,q\})= \partial_s A + \hbar_k^{-1}Q_k (\{p,q\}) \;.$$
Thus
$$\partial_t (B-C) = \hbar_k^{-2}\left(-i[Q_k(p)Q_k(q)] - \hbar_k Q_k (\{p,q\}\right)= \bigo(1/k)\;,$$
by bracket correspondence (P2).
Observe that $\partial_s U(0,s)=0$, so $B(0,s)=0$.
Further, $q(0,s) =0$, so $C(0,s)=0$. Thus
$$\|B(1,s) -C(1,s)\|_{\op{op}} = \bigo(1/k)\;.$$
But $C(1,s) = 0$ since $q(1,s)=0$. Thus
$\|B(1,s)\|_{\op{op}} = \bigo(1/k)$ and hence by \eqref{eq-U}
$$\|U(1,1)-U(1,0)\|_{\op{op}} = \bigo(1/k)\;,$$
and \eqref{eq-homotopy} follows. This proves item (i) of the theorem.

\medskip
\noindent
Let's analyze the quantization of the product of two Hamiltonian paths.
Let $\phi_t$ and $\psi_t$ be two paths generated by normalized Hamiltonians
$f_t$ and $g_t$ respectively, and denote $\theta_t=\phi_t\psi_t$.   Consider the corresponding Schroedinger evolutions
$$\dot{U}_k = -i\hbar_k^{-1}Q_k(f_t)U_k\;\;, U_k(0) =\Id\;,$$
$$\dot{V}_k = -i\hbar_k^{-1}Q_k(g_t)V_k\;\;, V_k(0) =\Id\;.$$
Put $$S(t) = Q_k(f_t) + U_k(t)Q_k(g_t)U_k(t)^{-1}, \;\;\; W_k(t) = U_k(t)V_k(t)\;.$$
Observe that
\begin{equation}\label{eq-W}
\dot{W_k} = -i\hbar_k^{-1}S(t)W\;.
\end{equation}
Since $\theta_t$ is generated by
$h_t:=f_t + g_t \circ \phi_t^{-1}$,
the Egorov theorem (Theorem \ref{thm-Egorov} ) yields
$$Q_k (h_t) = S(t) + \bigo(1/k^2)\;.$$
Denote by $Z_k(t)$ the Schroedinger evolution of $\theta_t$, that is
$$\dot{Z}_k = -i\hbar_k^{-1}Q_k(h_t)Z_k = (-i\hbar_k^{-1}S(t)+ \bigo(1/k))Z_k\;\;, Z_k(0) =\Id\;.$$
Comparing this equation with \eqref{eq-W} we conclude that
$$\|U_k(1)V_k(1) - Z_k(1)\|_{\op{op}} = \bigo(1/k)\;.$$
Thus $\mu_k$ is an almost-representation, which proves item (ii) of the theorem.

\medskip
\noindent
Finally, assume that a Hamiltonian $f_t$ generates a Hamiltonian path $\phi_t$ with
$\phi_1 \neq \Id$. Thus $\phi_1$ displaces an open set $Y \subset M$:
$\phi_1 (Y) \cap Y =\emptyset$. Take a non-vanishing function $g$ supported in $\phi_1(Y)$. Observe that
\begin{equation}
\label{eq-gg}
\|g \circ \phi^{-1} - g\|=\|g\|\;.
\end{equation}
Put $A_k:= Q_k(g)$.
Let $U_k$ be the unitary operator quantizing $\phi_1$. By the Egorov theorem,
$Q_k(g \circ \phi^{-1}) = U_kA_kU_k^{-1} + \bigo(1/k^2)$.
It follows from \eqref{eq-gg} and (P1) that
$\|U_kA_kU_k^{-1} - A\|_{\op{op}} = \|A\|_{\op{op}}+\bigo(1/k)$.
Estimating
$$\|A\|_{\op{op}}+\bigo(1/k) = \|U_kA_kU_k^*-A\|_{\op{op}} = $$
$$ \|U_kAU_k^*-U_kA+U_kA-A\|_{\op{op}}\leq 2\|A\|_{\op{op}} \cdot\|\Id-U_k\|_{\op{op}}\;,$$
we get that $\|\Id-U_k\|_{\op{op}} \geq 1/2+\bigo(1/k)$,
which proves item (iii) of the theorem.
\qed

\begin{rem}\label{rem-proj-vsp}{\rm Replacing $U_k$ by $e^{i\theta} U_k$ in
    the proof of (iii), we get that $$\|U_k -e^{i\theta}\Id\|_{\op{op}}
    \geq 1/2+\bigo(1/k)$$ for every phase $\theta$. This implies that the approximate projective representation $\nu_k$ appearing right after Theorem \ref{thm-fine-1} satisfies, for every $\phi \in \Ham(M,\omega)$,
$$\delta_p(\nu_k(\phi),\Id) \geq \text{const} >0, \;\; \forall k \in \N\;,$$
provided $\phi \neq \Id$.}
\end{rem}

\section{Loop quantization}\label{sec-loop}
In this section we prove Theorem \ref{thm-fine-1} from the introduction.
A more detailed formulation of this result appears in Theorem \ref{th:loop} below.

\subsection{Action and Maslov index}
Let $(M, \om)$ be a compact symplectic manifold equipped with a prequantum
line bundle $L$ and an auxiliary line bundle $L'$ such that
\begin{gather} \label{eq:hyp}
 c_1^{\R} (L') = \la c_1^{\R} (L) + \tfrac{1}{2} c_1^{\R} ( K)
\end{gather}
where
$K$ is the canonical line bundle.

Since $\frac{1}{i}\om$ is the curvature of $L$, the periods of $\om$ are
multiple of $2 \pi$, so the action of any contractible periodic trajectory
$\ga(t)$, $t \in [0,T]$ of a Hamiltonian
$(H_t)$ is well-defined modulo $2 \pi \Z$ and given by the usual formula
\begin{gather} \label{eq:defaction}
A (\ga)  = \textstyle{\int}_D \om - \int_0^T H_t ( \ga (t) ) dt
\end{gather}
where $D$ is a disc with boundary $\ga$.
We can even
define the action modulo $2\pi$ of any periodic trajectory, by using parallel
transport in $L$ instead of the integral of $\om$.

If $(H_t)$ generates a loop $\mathcal{L} = ( \phi_t, \, t \in [0,1])$ of Hamiltonian diffeomorphisms, then our assumption on $L'$ allows to define a mixed
action-Maslov invariant as follows \cite{P-loop}. By Floer theory, any trajectory
$  \phi_t(x)$, $t \in [0,1]$ is the boundary of a disc $D$.
We set
\begin{gather}
\label{eq:invmaslovaction}
I ( \mathcal{L} ) = \la \bigl( \textstyle{\int}_D \om - \int_0^1 H_t ( \phi_t (x) )
\, dt \bigr)  + \pi m  (
\psi)
\end{gather}
where $\psi$ is the loop of $\op{Sp} (2n)$ obtained by trivialising
the symplectic bundle $TM$ over $D$ and defining $\psi (t) := T_{x} \phi_t $,
$m ( \psi) = 0$ or $1$  according to the class of $\psi$ in $\pi_1 ( \op{Sp} (2 n)) =\Z$ is
even or odd. One readily checks that $I(\mathcal{L})$ is well defined modulo $2 \pi \Z$.

\subsection{Quantization of  a Hamiltonian loop}

Assume now that $(M, \om)$ is K\"ahler, that $L$ and $L'$ are
holomorphic hermitian line bundles with Chern curvatures $\Theta$ and $\Theta'$
satisfying $\Theta = \frac{1}{i} \om$, $\Theta' = \la \Theta + \frac{1}{2}
\Theta_K$. Consider the space  $\Hilb_k$  of holomorphic
sections of $L^k \otimes L'$. For any $f \in \Ci ( M , \R)$, we define
the operator $\Op_k  ( f)$ as in \eqref{eq:def_op_k}

Let $(H_t)$ be a Hamiltonian of $M$ generating a loop
$\mathcal{L} = (\phi_t, \, t \in [0,1])$.  Introduce the quantum propagator
$U_{t,k}$,
$$ \frac{1}{i(k+ \la)} \partial_t U_{k,t} + \Op_k ( H_t) U_{k,t} = 0 , \qquad
U_{k,0} = \Id$$
We assume from now on that $M$ is connected, so the periodic trajectories $(\phi_t(x), \, t \in
[0,1])$ have all the same action, denoted by $A( \mathcal{L})$.

\begin{thm} \label{th:loop}
  We have $U_{k,1} = e^{ik A ( \mathcal{L} ) + i I( \mathcal{L})} +
  \bigo ( k^{-1})$.
 \end{thm}

 \begin{proof}
   We can rewrite the Schr\"odinger equation as
   $$ \tfrac{1}{ik} \partial_t U_{k,t} + (1 + \tfrac{\la}{k}) \Op_k ( H_t) U_{k,t} =
   0 $$
  Then, by \cite[Theorem 4.2]{CLF}  the Schwartz kernel of $U_{k,t}$ is a Lagrangian state associated to
the graph of $\phi_t$. We refer to \cite{CLF} for the precise definitions.
What is important to us here is that since $\phi_1$ is the identity,
\begin{gather} \label{eq:u_1}
U_{k,1} = e^{i k
  \theta} T_k (\si)  + \bigo ( k^{-1})
\end{gather}
where $\theta$ is a real number, $\si \in \Ci (M)$ and
$T_k (\si )$ is the Berezin-Toeplitz operator with multiplicator $\si$ defined
as in section \ref{sec:an-impr-corr}.

Furthermore, we can compute $\theta$ and
$\si$ by introducing a half-form bundle (i.e., the square root of the canonical bundle)
denoted by $\delta$.
It is possible that such a
bundle does not exist on $M$ but we only need it on the trajectory $\ga$ of a
given point $x$. In this case we take a disk $D$ with boundary $\ga$ and choose the
square root $\delta$ which extends to $D$.

Then by \cite[Theorem 1.1]{CLF}
$$ U_{k,t} ( \phi_t (x) , x) \sim \Bigl( \frac{k}{2\pi} \Bigr)^n
 e^{ \frac{1}{i}
 \int_0^{t} H^{\op{sub}}_{r} ( \phi_r(x)) \; dr} \Bigr[ \phi_{t}^L (x) \Bigl]^{\otimes k} \otimes \mathcal{T}_{t}^{L_1}
(x)  \otimes  \bigl[ \mathcal{D}_t (x) \bigr]^{1/2}\;.
$$
Here $\phi_t^L $ is the prequantum lift of $\phi_t$ to $L$, and
$H_r^{\op{sub}} =\la H_t$ is the subprincipal symbol of $(1 + \frac{\la}{k})
 \Op_k ( H_t)$. The second term $\mathcal{T}_t^{L_1}(x) :L_1|_x \rightarrow L_1|_{\phi_t(x)}$
 is the parallel transport in the line bundle $L_1 = L' \otimes \delta^{-1}$.
 It is the multiplication by $\exp ( i \la  \int_D \om )$ because the
 curvature of $L_1$ is $\Theta' - \frac{1}{2} \Theta_K = \la \Theta =
 \frac{\la}{i} \om$. The last term is the square root of an isomorphism
 $\mathcal{D}_t(x) : K_x \rightarrow K_{\phi_t(x)}$ defined by
 $$ \mathcal{D}_t (x) ( \al ) ( (T_x\phi_t)^{1,0} u ) = \al (u) , \qquad
 \forall \al \in K_x, \, u \in \det T_x^{1,0} M\;. $$
 Here the square root is chosen so as to be continuous and equal to $1$ at $t=0$.

 On the other hand, by \eqref{eq:u_1}, $U_{k,1} (x,x) = \bigl( k/2 \pi \bigr)^n
e^{i k \theta} ( \si (x) + \bigo ( k^{-1} )) $.  Now $\phi_1^L (x) =
e^{ i A( \mathcal{L}) }$ implies that $\theta = A( \mathcal{L})$ and it remains to prove that
\begin{gather} \label{eq:toprove}
e^{ \frac{1}{i}
 \int_0^{1} H^{\op{sub}}_{r} ( \phi_r(x)) \; dr}  \mathcal{T}_{1}^{L_1}
(x)  \otimes  \bigl[ \mathcal{D}_1 (x) \bigr]^{1/2} = e^{i I( \mathcal{L} )}
\end{gather}
Since $T_x \phi_1 $ is the identity of $T_xM$, $\mathcal{D}_ 1(x) $ is the identity of
 $K_x$ so $$( \mathcal{D}_t(x))^{1/2} = \pm \Id_{\delta_x}\;.$$ To determine
 the sign,  we trivialize $TM$ along $\ga$ with an symplectic frame, so
 that $(T_x \phi_t)$ becomes a loop $\alpha$ of symplectic matrices based at the
 identity and in the corresponding trivialisation of $K$, $\mathcal{D}_t(x)$
 is the multiplication by a complex number. The sign we search depends
 only on the homotopy class of $\alpha$.  Since $\op{Sp} (2n)$
 deformation retracts to its subgroup $\op{U}(n)$, we can assume that $\alpha$ is a
 loop of $\op{U}(n)$, in which case $\mathcal{D}_t (x) $ is the complex determinant
 of $\alpha (t)$. Thus, our sign is positive or negative according to the class of $\alpha$ in $\pi_1 (
 \op{Sp} (2n))=\Z$ is even or odd. We conclude that each factor in \eqref{eq:toprove}
 corresponds to a summand in \eqref{eq:invmaslovaction}, which completes the proof.
\end{proof}

\medskip
\noindent {\bf Acknowledgments.} We are grateful to Alex Lubotzky for his help
with the class of groups $\cP\cL\cO_p$, and for valuable comments on
\cite{DGLT,LO}. We thank David Fisher for useful discussions.

\vspace{2cm}

\noindent
\begin{tabular}{ll}
{\bf Laurent Charles} &  {\bf Leonid Polterovich} \\
Sorbonne Universit\'e, Universit\'e de Paris, CNRS &  Tel Aviv University \\
Institut de Math\'ematiques de Jussieu-Paris Rive Gauche & School of Mathematical Sciences\\
F-75005 Paris, France  & 69978, Tel Aviv, Israel\\
{\em E-mail:} \texttt{laurent.charles@imj-prg.fr} & \texttt{polterov@tauex.tau.ac.il}
\end{tabular}

\end{document}